\pdfoutput=1
\documentclass[11pt]{article}
\usepackage{subfig}
\usepackage{amssymb,amsmath,amsthm}
\usepackage{graphics,epsfig}
\usepackage{algorithmic}
\usepackage{algorithm}
\usepackage[margin=1 in]{geometry}

\newtheorem{lem}{Lemma}[section]

\newtheorem{thm}[lem]{Theorem}

\title{Scheduling Coflows with Precedence Constraints for Minimizing the Total Weighted Completion Time in Identical Parallel Networks}
\author{Chi-Yeh~Chen 
\\ Department of Computer Science and Information
Engineering, \\ National Cheng Kung University, \\
Taiwan, ROC. \\
chency@csie.ncku.edu.tw.}

\begin{document}

\maketitle
\begin{abstract}
Coflow is a recently proposed network abstraction for data-parallel computing applications. This paper considers scheduling coflows with precedence constraints in identical parallel networks, such as to minimize the total weighted completion time of coflows. The identical parallel network is an architecture based on multiple network cores running in parallel. In the divisible coflow scheduling problem, the proposed algorithm achieves $(6-\frac{2}{m})\mu$ and $(5-\frac{2}{m})\mu$ approximate ratios for arbitrary release time and zero release time, respectively, where $m$ is the number of network cores and $\mu$ is the coflow number of the longest path in the precedence graph. In the indivisible coflow scheduling problem, the proposed algorithm achieves $(4m+1)\mu$ and $4m\mu$ approximate ratios for arbitrary release time and zero release time, respectively. In the single network core scheduling problem, we propose a $5\mu$-approximation algorithm with arbitrary release times, and a $4\mu$-approximation without release time. Moreover, the proposed algorithm can be modified to solve the coflows of multi-stage jobs scheduling problem. In multi-stage jobs, coflow is transferred between servers to enable starting of next stage. This means that there are precedence constraints between coflows of job. Our result represents an improvement upon the previous best approximation ratio of $O(\tilde{\mu} \log(N)/ \log(\log(N)))$ where $\tilde{\mu}$ is the maximum number of coflows in a job and $N$ is the number of servers.

\end{abstract}

\section{Introduction}\label{sec:Introduction}
Due to growing computing demands, large data centers have become critical to cloud computing.
In large data centers, the benefits of application-aware network scheduling have been demonstrated by structured traffic patterns for distributed applications~\cite{Chowdhury2014, Chowdhury2015, Zhang2016, Agarwal2018}.
Many data-parallel computing applications such as MapReduce~\cite{Dean2008}, Hadoop~\cite{Shvachko2010, borthakur2007hadoop}, Dryad~\cite{isard2007dryad} and Spark~\cite{zaharia2010spark} have been successfully used by many users. Therefore, this has led to a proliferation of related applications~\cite {dogar2014decentralized, chowdhury2011managing}.
Data-parallel computing applications generate a lot of intermediate data (flows) during the computing stage.
These flows need to be transmitted through different machines during the communication stage for further processing.
Due to the large amount of data flow transmission requirements caused by numerous applications, the data center must have sufficient data transmission and scheduling capabilities.
Looking at the communication patterns of data-parallel computing applications, the interaction of all flows between two sets of machines becomes important.
This collective communication pattern in the data center is abstracted by coflow traffic~\cite{Chowdhury2012}.

This paper considers an architecture based on multiple identical network cores running in parallel. The goal is to schedule coflows with precedence constraints in the identical parallel networks such that the total weighted coflow completion time is minimized. 
In the identical parallel networks, the coflow can be considered as divisible and indivisible.
In the divisible coflow scheduling problem, the flows in a coflow can be distributed in different network cores. 
However, data in a flow is only distributed to the same core. 
In the indivisible coflow scheduling problem, flows in a coflow can only be distributed in the same network core.
This paper then modifies the proposed algorithm to solve the multi-stage jobs scheduling problem. 
A multi-stage job contains multiple coflows with precedence constraints.
Scheduling the coflows of multi-stage job problem requires minimizing the total weighted completion time and makespan of the job.

\subsection{Related Work}
Chowdhury and Stoica~\cite{Chowdhury2012} first introduced the coflow abstraction to describe communication patterns in data centers.
The coflow scheduling problem has been shown to be strongly NP-hard by a reduction of the well-studied concurrent open shop scheduling problem~\cite{chen2007supply, garg2007order, leung2007scheduling, mastrolilli2010minimizing, wang2007customer}.
Therefore, this problem requires efficient approximate algorithms rather than exact algorithms.
The coflow scheduling problem is NP-hard to approximate within a factor better than $2-\epsilon$ due to the inapproximability of concurrent open shop problem~\cite{Sachdeva2013, shafiee2018improved, ahmadi2020scheduling, Bansal2010, Sachdeva2013}.
Since the coflow abstraction was proposed, many relevant investigations have been carried out for scheduling coflows~\cite{Chowdhury2014, Chowdhury2015, Qiu2015, zhao2015rapier, shafiee2018improved, ahmadi2020scheduling}.
Qiu~\textit{et al.}~\cite{Qiu2015} proposed the first polynomial-time deterministic approximation algorithm.
Since then, the best approximation ratio achievable within polynomial time was improved from $\frac{64}{3}$ to 4 when coflow is released at zero, and from $\frac{76}{3}$ to 5 when coflow is released at arbitrary time~\cite{Qiu2015, khuller2016brief, shafiee2018improved, ahmadi2020scheduling}.
In the single coflow scheduling problem in a heterogeneous parallel network, Huang \textit{et al.}~\cite{Huang2020} proposed an $O(m)$-approximation algorithm, where $m$ is the number of network cores. When each job has multiple coflows with precedence constraints, Shafiee and Ghaderi~\cite{shafiee2021scheduling} proposed a polynomial-time algorithm with approximation ratio of $O(\tilde{\mu} \log(N)/\log(\log(N)))$, where $\tilde{\mu}$ is the maximum number of coflows in a job and $N$ is the number of servers.

\subsection{Our Contributions}
This paper considers the scheduling coflows with precedence constraints problem in identical parallel networks. Our results are as follows:
\begin{itemize}
\item In the divisible coflow scheduling problem, we first propose a $(6-\frac{2}{m})\mu$-approximation algorithm with arbitrary release times, and a $(5-\frac{2}{m})\mu$-approximation without release time, where $m$ is the number of network cores and $\mu$ is the coflow number of the longest path in the precedence graph. 

\item When coflow is indivisible, we propose a $(4m+1)\mu$-approximation algorithm with arbitrary release times, and a $4m\mu$-approximation without release time.
\item In the single network core scheduling problem, we propose a $5\mu$-approximation algorithm with arbitrary release times, and a $4\mu$-approximation without release time.
\item In the coflows of multi-stage jobs scheduling problem, our algorithms achieve $O(\mu)$ approximate ratio for minimizing the total weighted completion time and makespan.
\end{itemize}
\subsection{Organization}
The rest of this article is organized as follows. Section \ref{sec:Preliminaries} introduces basic notations and preliminaries. Section \ref{sec:Algorithm1} presents an algorithm for divisible coflow scheduling. Section \ref{sec:Algorithm2} presents an algorithm for indivisible coflow scheduling. Section \ref{sec:Algorithm3} presents an algorithm for the single network core scheduling. Section \ref{sec:Algorithm4} modifies our algorithms to solve the coflows of multi-stage jobs scheduling problem. Section \ref{sec:Conclusion} draws conclusions.

\section{Notation and Preliminaries}\label{sec:Preliminaries}
Given a set $M$ of $N \times N$ non-blocking switches and a set $\mathcal{K}$ of coflows, the scheduling coflows with precedence constraints problem asks for a minimum total weighted coflow completion time. This paper considers the identical parallel network which is an architecture based on multiple identical network cores (or switch) running in parallel. Let $m$ be the number of network cores. Each network core has $N$ input links connected to $N$ source servers and $N$ output links connected to $N$ destination servers. For each network core, the $i$-th input (or $j$th output) port is connected to the $i$-th source server (or $j$-th destination server).
Therefore, each source server (or destination server) has $m$ simultaneous uplinks (or downlinks). Each uplink (or downlink) can be a bundle of multiple physical links in the actual topology~\cite{Huang2020}. 
Let $\mathcal{I}$ be the source server set and $\mathcal{J}$ be the destination server set. 
For simplicity, we assume that all links in each network core have the same capacity (or the same speed).

A coflow is defined as a set of independent flows whose completion time is determined by the completion time of the latest flow in the set.
We can use a $N\times N$ demand matrix $D^{(k)}=\left(d_{ijk}\right)_{i,j=1}^{N}$ to express the coflow $k$ where $d_{ijk}$ denote the size of the flow to be transferred from input $i$ to output $j$ in coflow $k$. We also can use a triple $(i, j, k)$ to express a flow in which $i \in \mathcal{I}$ is its source server and $j \in \mathcal{J}$ is its destination server, $k$ is the coflow to which it belongs.
For simplicity, we assume that all flows in a coflow arrive at the system at the same time. We also assume that flows consist of discrete data units, so their sizes are integers (as shown in~\cite{Qiu2015}). 

Let $\mathcal{K}$ be a set of coflows and let $r_k$ be the released time of coflow $k$ for $k = 1, 2, \ldots, |\mathcal{K}|$. 
The completion time of coflow $k$ is denoted by $C_k$. Let $w_{k}$ be the weight of coflow $k$.
There are precedence constraints among the coflows which are represented by a directed acyclic graph (DAG) $G=(\mathcal{K}, E)$.
For the case of arc $(k', k)\in E$ and $k', k\in \mathcal{K}$, all flows of coflow $k'$ should be done before we start scheduling any flow of coflow $k$.
In this case, we say coflow $k'$ precedes coflow $k$, and denote it by $k'\prec k$. Let $\mu$ be the coflow number of the longest path in the DAG.
The goal is to schedule coflows with precedence constraints in an identical parallel network such that the total weighted coflow completion time $\sum_{k\in \mathcal{K}} w_{k}C_{k}$ is minimized. 
Table~\ref{tab:notations} presents the notation and terminology that are used herein.
\begin{table}[ht]
\caption{Notation and Terminology}
\vspace{2mm}
    \centering
        \begin{tabular}{||c|p{4in}||}
    \hline
     $m$      & The number of network cores.          \\
    \hline    
     $N$      & The number of input/output ports.         \\
    \hline
     $\mathcal{I}, \mathcal{J}$ & The source server set and the destination server set.         \\
    \hline    
     $\mathcal{K}$ & The set of coflows.         \\
    \hline
     $D^{(k)}$     & The demand matrix of coflow $k$. \\
    \hline    
     $d_{ijk}$     & The size of the flow to be transferred from input $i$ to output $j$ in coflow $k$.   \\
    \hline     
     $C_k$     & The completion time of coflow $k$.   \\
    \hline     
     $C_{ijk}$ & The completion time of flow $(i, j, k)$. \\
    \hline     
     $r_k$     & The released time of coflow $k$. \\
    \hline     
     $w_{k}$   & The weight of coflow $k$. \\
    \hline     
		$\mu$      & The coflow number of the longest path in the DAG. \\
		\hline 
		 $\bar{C}_1, \ldots, \bar{C}_n$ & An optimal solution to the linear program. \\
		\hline 
		 $\hat{C}_{1}, \ldots, \hat{C}_{n}$ & The length of time interval between ready time and completion time for each coflow. \\
		\hline 		
		$\tilde{C}_{1}, \ldots, \tilde{C}_{n}$ & The schedule solution to our algorithm. \\
		\hline 
        \end{tabular}
    \label{tab:notations}
\end{table}


\section{Approximation Algorithm for Divisible Coflow Scheduling}\label{sec:Algorithm1}
This section consideres the divisible coflow scheduling problem.  
The different flows in a divisible coflow can be transmitted through different cores. However, data in a flow is only distributed to the same core. 
Let $\mathcal{K}_{i}=\left\{(k, j)| d_{ijk}>0, \forall k\in \mathcal{K}, \forall j\in \mathcal{J} \right\}$ be the set of flows in which the flows need to transmit through input port $i$.
Let $\mathcal{K}_{j}=\left\{(k, i)| d_{ijk}>0, \forall k\in \mathcal{K}, \forall i\in \mathcal{I} \right\}$ be the set of flows in which the flows need to transmit through output port $j$.
For any subset $S_{i}\subseteq \mathcal{K}_{i}$ (or $S_{j}\subseteq \mathcal{K}_{j}$), let $d(S_{i})=\sum_{(k, j)\in S_{i}} d_{ijk}$ (or $d(S_{j})=\sum_{(k, i)\in S_{j}} d_{ijk}$) and $d^2(S_{i})=\sum_{(k, j)\in S_{i}} d_{ijk}^{2}$ (or $d^2(S_{j})=\sum_{(k, i)\in S_{j}} d_{ijk}^{2}$). 
Let $C_{k}$ be the completion time of coflow $k$ and let $C_{ijk}$ be the completion time of flow $(i, j, k)$.
We can formulate our problem as the following linear programming relaxation:
\begin{subequations}\label{coflow:main}
\begin{align}
& \text{min}  && \sum_{k \in \mathcal{K}} w_{k} C_{k}     &   & \tag{\ref{coflow:main}} \\
& \text{s.t.} && C_{k} \geq C_{ijk}, && \forall k\in \mathcal{K}, \forall i\in \mathcal{I}, \forall j\in \mathcal{J} \label{coflow:a} \\
&  && C_{ijk}\geq r_k+d_{ijk}, && \forall k\in \mathcal{K}, \forall i\in \mathcal{I}, \forall j\in \mathcal{J} \label{coflow:b} \\
&  && C_{ijk}\geq C_{k'}+d_{ijk}, && \forall k, k'\in \mathcal{K}, \forall i\in \mathcal{I}, \forall j\in \mathcal{J}: k'\prec k\label{coflow:e} \\
&  && \sum_{(k, j)\in S_{i}}d_{ijk}C_{ijk}\geq \frac{1}{2m} \left(d(S_{i})^2+d^2(S_{i})\right),&& \forall i\in \mathcal{I}, \forall S_{i}\subseteq \mathcal{K}_{i} \label{coflow:c} \\
&  && \sum_{(k, i)\in S_{j}}d_{ijk}C_{ijk} \geq \frac{1}{2m} \left(d(S_{j})^2+d^2(S_{j})\right),&& \forall j\in \mathcal{J}, \forall S_{j}\subseteq \mathcal{K}_{j} \label{coflow:d} 
\end{align}
\end{subequations}

The constraint~(\ref{coflow:a}) ensures that the completion time of coflow $k$ is bounded by all its flows. 
The constraint~(\ref{coflow:b}) ensures that the completion time of any flow $(i, j, k)$ is at least its release time $r_k$ plus its load. 
The constraint~(\ref{coflow:e}) is the precedence constraint.
The constraints~(\ref{coflow:c}) and (\ref{coflow:d}) are modified from the scheduling literature~\cite{Leslie1997} to lower bound the completion time variable in the input port and the output port respectively. 

Our algorithm flow-driven-list-scheduling (described in Algorithm~\ref{Alg1}) is modified from our previous algorithm~\cite{Chency2022}.
Given $n$ flows from all coflows, we have an optimal solution $\bar{C}_1, \ldots, \bar{C}_n$ from the linear program (\ref{coflow:main}).
Without loss of generality, we assume $\bar{C}_{1}\leq \cdots\leq \bar{C}_{n}$ and schedule the flows iteratively in the order of this list. 
For each flow $f$, we find a network core $h$ such that the complete time of $f$ in this network core is minimized (lines 5-9). 
A flow is called to be \textit{ready} for scheduling if all predecessors of the flow have been fully transmitted.
Lines 10-22 schedule every ready, released and incomplete flows. This procedure (lines 10-22) is modified from Shafiee and Ghaderi's algorithm~\cite{shafiee2018improved}.

\begin{algorithm}
\caption{flow-driven-list-scheduling}
    \begin{algorithmic}[1]
		    \REQUIRE a vector $\bar{C}\in \mathbb{R}_{\scriptscriptstyle \geq 0}^{n}$ used to decide the order of scheduling
				\STATE let $load_{I}(i,h)$ be the load on the $i$-th input port of the network core $h$
				\STATE let $load_{O}(j,h)$ be the load on the $j$-th output port of the network core $h$
				\STATE let $\mathcal{A}_h$ be the set of flows allocated to network core $h$
				\STATE both $load_{I}$ and $load_{O}$ are initialized to zero and $\mathcal{A}_h=\emptyset$ for all $h\in [1, m]$
				\FOR{every flow $f=(i, j, k)$ in non-decreasing order of $\bar{C}_f$, breaking ties arbitrarily}
				    \STATE $h^*=\arg \min_{h\in [1, m]}\left(load_{I}(i,h)+load_{O}(j,h)\right)$
						\STATE $\mathcal{A}_{h^*}=\mathcal{A}_{h^*}\cup \left\{f\right\}$
						\STATE $load_{I}(i,h^*)=load_{I}(i,h^*)+d_f$ and $load_{O}(j,h^*)=load_{O}(j,h^*)+d_f$
				\ENDFOR
				\FOR{each $h\in [1, m]$ do in parallel}
				    \STATE wait until the first coflow is released
						\WHILE{there is some incomplete flow}
                \FOR{every ready, released and incomplete flow $f=(i, j, k)\in \mathcal{A}_{h}$ in non-decreasing order of $\bar{C}_f$, breaking ties arbitrarily}
										\IF{the link $(i, j)$ is idle}
										    \STATE schedule flow $f$
										\ENDIF
								\ENDFOR
								\WHILE{no new flow is ready, completed or released}
								    \STATE transmit the flows that get scheduled in line 15 at maximum rate 1.
								\ENDWHILE
						\ENDWHILE
				\ENDFOR
   \end{algorithmic}
\label{Alg1}
\end{algorithm}

\subsection{Analysis}
This section first shows that when the precedence constraint is omitted, the proposed algorithm achieves an approximation ratio of $r$, where $r=6-\frac{2}{m}$ for arbitrary release times and $r=5-\frac{2}{m}$ for zero release time. 
This means that any coflow $k$ takes at most $r\bar{C}_k$ to transmit after all its predecessors have been fully transmitted where $\bar{C}_k$ is an optimal solution to the linear program (\ref{coflow:main}).
Therefore, the proposed algorithm achieves an approximation ratio of $(6-\frac{2}{m})\mu$ with arbitrary release times, and an approximation ratio of $(5-\frac{2}{m})\mu$ without release time where $\mu$ is the coflow number of the longest path in the precedence graph. 
Let $\bar{C}_{1}\leq \cdots\leq \bar{C}_{n}$ be an optimal solution to the linear program (\ref{coflow:main}), and let $\hat{C}_{1}, \ldots, \hat{C}_{n}$ be the length of time interval between ready time and completion time for each coflow. Moreover, let $\tilde{C}_{1}, \ldots, \tilde{C}_{n}$ denote the completion times in the schedule found by flow-driven-list-scheduling.
According to our previous results~\cite{Chency2022}, we have the following lemma:

\begin{lem}\label{lem:lem3}
\cite{Chency2022} For each coflow $k=1, \ldots, n$,
\begin{eqnarray*}
\hat{C}_{k}\leq \left(6-\frac{2}{m}\right)\bar{C}_{k}.
\end{eqnarray*}
\end{lem}

According to lemma~\ref{lem:lem3}, we have the following theorem:
\begin{thm}\label{thm:thm1}
The flow-driven-list-scheduling has an approximation ratio of, at most, $\left(6-\frac{2}{m}\right)\mu$.
\end{thm}
\begin{proof}
For each coflow $k=1, \ldots, n$, the completion time is at most the time of the longest path in the precedence graph at time 0.
Assume the longest path of coflow $k$ is $v_1v_2\cdots v_f$ where $v_f=k$.
We have
\begin{eqnarray}
\tilde{C}_{k} & \leq & \sum_{q=1}^{f} \hat{C}_{v_q} \label{lem4:eq1}\\
              & \leq & \sum_{q=1}^{f} \left(6-\frac{2}{m}\right)\bar{C}_{v_q} \label{lem4:eq2}\\
							& \leq & \sum_{q=1}^{f} \left(6-\frac{2}{m}\right)\bar{C}_{k} \label{lem4:eq3}\\
							& =    & f \left(6-\frac{2}{m}\right)\bar{C}_{k} \label{lem4:eq4}\\
							& \leq & \mu\left(6-\frac{2}{m}\right)\bar{C}_{k}. \label{lem4:eq5}
\end{eqnarray}
The inequality~(\ref{lem4:eq1}) is that the bound of $\tilde{C}_{k}$ is the sum of all length of of time intervals in the longest path.
The inequality~(\ref{lem4:eq2}) is based on lemma~\ref{lem:lem3}.
The inequality~(\ref{lem4:eq3}) is due to the constraint~(\ref{coflow:e}) in linear program~(\ref{coflow:main}).
The inequality~(\ref{lem4:eq5}) is because $\mu$ is the coflow number of the longest path in the precedence graph.
\end{proof}

When all coflows are release at time zero, we have the following lemma:
\begin{lem}\label{lem:lem4}
\cite{Chency2022} For each coflow $k=1, \ldots, n$,
\begin{eqnarray*}
\hat{C}_{k}\leq \left(5-\frac{2}{m}\right)\bar{C}_{k}.
\end{eqnarray*}
when all coflows are released at time zero.
\end{lem}

According to lemma~\ref{lem:lem4}, we have the following theorem:
\begin{thm}\label{thm:thm2}
For the special case when all coflows are released at time zero, the flow-driven-list-scheduling has an approximation ratio of, at most, $\left(5-\frac{2}{m}\right)\mu$.
\end{thm}
\begin{proof}
The proof is similar to that of theorem~\ref{thm:thm1}.
\end{proof}

\section{Approximation Algorithm for Indivisible Coflow Scheduling}\label{sec:Algorithm2}
This section consider the indivisible coflow scheduling problem. We assume all flows in a coflow can only be transmitted through the same core. For every coflow $k$ and input port $i$, let $L_{ik}=\sum_{j=1}^{N}d_{ijk}$ be the total amount of data that coflow $k$ needs to transmit through input port $i$. Moreover, let $L_{jk}=\sum_{i=1}^{N}d_{ijk}$ be the total amount of data that coflow $k$ needs to transmit through output port $j$.
We can formulate our problem as the following linear programming relaxation:
\begin{subequations}\label{incoflow:main}
\begin{align}
& \text{min}  && \sum_{k \in \mathcal{K}} w_{k} C_{k}     &   & \tag{\ref{incoflow:main}} \\
& \text{s.t.} && C_{k}\geq r_k+L_{ik}, && \forall k\in \mathcal{K}, \forall i\in \mathcal{I} \label{incoflow:a} \\
&             && C_{k}\geq r_k+L_{jk}, && \forall k\in \mathcal{K}, \forall j\in \mathcal{J} \label{incoflow:b} \\
&  && C_{k}\geq C_{k'}+L_{ik}, && \forall k, k'\in \mathcal{K}, \forall i\in \mathcal{I}: k'\prec k\label{incoflow:e} \\
&  && C_{k}\geq C_{k'}+L_{jk}, && \forall k, k'\in \mathcal{K}, \forall j\in \mathcal{J}: k'\prec k\label{incoflow:f} \\
&  && \sum_{k\in S}L_{ik}C_{k} \geq \frac{1}{2m} \left(\sum_{k\in S} L_{ik}^2+\left(\sum_{k\in S} L_{ik}\right)^2\right),&& \forall i\in \mathcal{I}, \forall S\subseteq \mathcal{K} \label{incoflow:c} \\
&  && \sum_{k\in S}L_{jk}C_{k} \geq \frac{1}{2m} \left(\sum_{k\in S} L_{jk}^2+\left(\sum_{k\in S} L_{jk}\right)^2\right),&& \forall j\in \mathcal{J}, \forall S\subseteq \mathcal{K} \label{incoflow:d} 
\end{align}
\end{subequations}

The constraints~(\ref{incoflow:a}) and (\ref{incoflow:b}) ensure that the completion time of any coflow $k$ is at least its release time $r_k$ plus its load. 
The constraints~(\ref{incoflow:e}) and (\ref{incoflow:f}) are the precedence constraints.
The constraints~(\ref{incoflow:c}) and (\ref{incoflow:d}) are modified from the scheduling literature~\cite{Leslie1997, ahmadi2020scheduling} to lower bound the completion time variable in the input port and the output port, respectively.

Our algorithm coflow-driven-list-scheduling (described in Algorithm~\ref{Alg2}) is modified from our previous algorithm~\cite{Chency2022}.
Given a set $\mathcal{K}$ of $n$ coflows, we have an optimal solution $\bar{C}_1, \ldots, \bar{C}_n$ from the linear program (\ref{incoflow:main}). 
Without loss of generality, we assume $\bar{C}_{1}\leq \cdots\leq \bar{C}_{n}$ and schedule all the flows iteratively in the order of this list.
For each coflow $k$, we find a network core $h$ such that the complete time of coflow $k$ in this network core is minimized (Lines 5-9). 
Lines 10-24 transmit all the flows in non-decreasing order in $\bar{C}$.

\begin{algorithm}
\caption{coflow-driven-list-scheduling}
    \begin{algorithmic}[1]
		    \REQUIRE a vector $\bar{C}\in \mathbb{R}_{\scriptscriptstyle \geq 0}^{n}$ used to decide the order of scheduling
				\STATE let $load_{I}(i,h)$ be the load on the $i$-th input port of the network core $h$
				\STATE let $load_{O}(j,h)$ be the load on the $j$-th output port of the network core $h$
				\STATE let $\mathcal{A}_h$ be the set of coflows allocated to network core $h$
				\STATE both $load_{I}$ and $load_{O}$ are initialized to zero and $\mathcal{A}_h=\emptyset$ for all $h\in [1, m]$
				\FOR{every coflow $k$ in non-decreasing order of $\bar{C}_k$, breaking ties arbitrarily}
				    \STATE $h^*=\arg \min_{h\in [1, m]}\left(\max_{i,j\in [1,N]}\left(load_{I}(i,h)+load_{O}(j,h)+L_{ik}+L_{jk}\right)\right)$
						\STATE $\mathcal{A}_{h^*}=\mathcal{A}_{h^*}\cup \left\{k\right\}$
						\STATE $load_{I}(i,h^*)=load_{I}(i,h^*)+L_{ik}$ and $load_{O}(j,h^*)=load_{O}(j,h^*)+L_{jk}$ for all $i,j\in [1,N]$
				\ENDFOR
				\FOR{each $h\in [1, m]$ do in parallel}
				    \STATE wait until the first coflow is released
						\WHILE{there is some incomplete flow}
						    \STATE for all $k\in \mathcal{A}_{h}$, list the ready, released and incomplete flows respecting the non-decreasing order in $\bar{C}_k$
								\STATE let $L$ be the set of flows in the list
                \FOR{every flow $f=(i, j, k)\in L$}
										\IF{the link $(i, j)$ is idle}
										    \STATE schedule flow $f$
										\ENDIF
								\ENDFOR
								\WHILE{no new flow is ready, completed or released}
								    \STATE transmit the flows that get scheduled in line 17 at maximum rate 1.
								\ENDWHILE
						\ENDWHILE
				\ENDFOR
   \end{algorithmic}
\label{Alg2}
\end{algorithm}

\subsection{Analysis}
This section first shows that when the precedence constraint is omitted, the proposed algorithm achieves an approximation ratio of $r$, where $r=4m+1$ for arbitrary release times and $r=4m$ for zero release time. 
According to our previous results~\cite{Chency2022}, we have the following lemma:

\begin{lem}\label{lem:lem33}
\cite{Chency2022} For each coflow $k=1, \ldots, n$,
\begin{eqnarray*}
\hat{C}_{k}\leq \left(4m+1\right)\bar{C}_{k}.
\end{eqnarray*}
\end{lem}

According to lemma~\ref{lem:lem33}, we have the following theorem:
\begin{thm}\label{thm:thm11}
The coflow-driven-list-scheduling has an approximation ratio of, at most, $\left(4m+1\right)\mu$.
\end{thm}
\begin{proof}
The proof is similar to that of theorem~\ref{thm:thm1}.
\end{proof}

We also have the following lemma:
\begin{lem}\label{lem:lem44}
\cite{Chency2022} For each coflow $k=1, \ldots, n$,
\begin{eqnarray*}
\hat{C}_{k}\leq 4m\bar{C}_{k}.
\end{eqnarray*}
when all coflows are released at time zero.
\end{lem}

According to lemma~\ref{lem:lem44}, we have the following theorem:
\begin{thm}\label{thm:thm22}
For the special case when all coflows are released at time zero, the coflow-driven-list-scheduling has an approximation ratio of, at most, $4m\mu$.
\end{thm}
\begin{proof}
The proof is similar to that of theorem~\ref{thm:thm1}.
\end{proof}


\section{Single Network Core Scheduling Problem}\label{sec:Algorithm3}
This section consider the single network core scheduling problem. The method is similar to the indivisible coflow scheduling problem in identical parallel networks.
We can formulate our problem as the following linear programming relaxation:

\begin{subequations}\label{incoflow:single:main}
\begin{align}
& \text{min}  && \sum_{k \in \mathcal{K}} w_{k} C_{k}     &   & \tag{\ref{incoflow:single:main}} \\
& \text{s.t.} && C_{k}\geq r_k+L_{ik}, && \forall k\in \mathcal{K}, \forall i\in \mathcal{I} \label{incoflow:single:a} \\
&             && C_{k}\geq r_k+L_{jk}, && \forall k\in \mathcal{K}, \forall j\in \mathcal{J} \label{incoflow:single:b} \\
&  && C_{k}\geq C_{k'}+L_{ik}, && \forall k, k'\in \mathcal{K}, \forall i\in \mathcal{I}: k'\prec k\label{incoflow:single:e} \\
&  && C_{k}\geq C_{k'}+L_{jk}, && \forall k, k'\in \mathcal{K}, \forall j\in \mathcal{J}: k'\prec k\label{incoflow:single:f} \\
&  && \sum_{k\in S}L_{ik}C_{k} \geq \frac{1}{2} \left(\sum_{k\in S} L_{ik}^2+\left(\sum_{k\in S} L_{ik}\right)^2\right),&& \forall i\in \mathcal{I}, \forall S\subseteq \mathcal{K} \label{incoflow:single:c} \\
&  && \sum_{k\in S}L_{jk}C_{k} \geq \frac{1}{2} \left(\sum_{k\in S} L_{jk}^2+\left(\sum_{k\in S} L_{jk}\right)^2\right),&& \forall j\in \mathcal{J}, \forall S\subseteq \mathcal{K} \label{incoflow:single:d} 
\end{align}
\end{subequations}

The constraints~(\ref{incoflow:single:a}) and (\ref{incoflow:b}) ensure that the completion time of any coflow $k$ is at least its release time $r_k$ plus its load. 
The constraints~(\ref{incoflow:single:e}) and (\ref{incoflow:single:f}) are the precedence constraints.
The constraints~(\ref{incoflow:single:c}) and (\ref{incoflow:single:d}) are used to lower bound the completion time variable in the input port and the output port, respectively. 

The algorithm is modified from Shafiee and Ghaderi's algorithm~\cite{shafiee2018improved}. Our algorithm single-network-core-list-scheduling (described in Algorithm~\ref{Alg3}) is as follows. Given $n$ flows from all coflows in the coflow set $\mathcal{K}$, we have an optimal solution $\bar{C}_1, \ldots, \bar{C}_n$ from the linear program (\ref{incoflow:single:main}). Without loss of generality, we assume $\bar{C}_{1}\leq \cdots\leq \bar{C}_{n}$ and schedule the flows iteratively in the order of this list. 

\begin{algorithm}
\caption{single-network-core-list-scheduling}
    \begin{algorithmic}[1]
		    \REQUIRE a vector $\bar{C}\in \mathbb{R}_{\scriptscriptstyle \geq 0}^{n}$ used to decide the order of scheduling
				    \STATE wait until the first coflow is released
						\WHILE{there is some incomplete flow}
               \FOR{every ready, released and incomplete flow $f=(i, j, k)$ in non-decreasing order of $\bar{C}_f$, breaking ties arbitrarily}
									\IF{the link $(i, j)$ is idle}
									    \STATE schedule flow $f$
									\ENDIF
							\ENDFOR
							\WHILE{no new flow is ready, completed or released}
							    \STATE transmit the flows that get scheduled in line 5 at maximum rate 1.
							\ENDWHILE
					\ENDWHILE
   \end{algorithmic}
\label{Alg3}
\end{algorithm}

\subsection{Analysis}

Based on the analysis in \cite{Leslie1997}, we have the following two lemmas: 
\begin{lem}\label{lem3:lem11}
For the $i$-th input with $n$ coflows, let $C_{1}, \ldots, C_{n}$ satisfy (\ref{incoflow:single:c}) and assume without loss of generality that $C_{1}\leq \cdots\leq C_{n}$. Then, for each $k=1,\ldots, n$, if $S=\left\{1, \ldots, k\right\}$,
\begin{eqnarray*}
C_{k}\geq \frac{1}{2} \sum_{f\in S} L_{if}.
\end{eqnarray*}
\end{lem}
\begin{proof}
According to (\ref{incoflow:single:c}) and the fact that $C_{f}\leq C_{k}$ for all $f\in S$, we have:
\begin{eqnarray*}
C_{k}\sum_{f\in S}L_{if} \geq \sum_{f\in S}L_{if}C_f \geq \frac{1}{2} \left(\sum_{f\in S} L_{if}^2+\left(\sum_{f\in S} L_{if}\right)^2\right) \geq \frac{1}{2} \left(\sum_{f\in S} L_{if}\right)^2.
\end{eqnarray*}
The following inequality can be obtained:
\begin{eqnarray*}
C_{k}\geq \frac{1}{2} \sum_{f\in S} L_{if}.
\end{eqnarray*}
\end{proof}

\begin{lem}\label{lem3:lem22}
For the $j$-th output with $n$ coflows, let $C_{1}, \ldots, C_{n}$ satisfy (\ref{incoflow:single:d}) and assume without loss of generality that $C_{1}\leq \cdots\leq C_{n}$. Then, for each $k=1,\ldots, n$, if $S=\left\{1, \ldots, k\right\}$,
\begin{eqnarray*}
C_{k}\geq \frac{1}{2} \sum_{f\in S} L_{jf}.
\end{eqnarray*}
\end{lem}
\begin{proof}
The proof is similar to that of lemma~\ref{lem3:lem11}.
\end{proof}

To find the bound of time interval length $\hat{C}_{k}$, we omit the precedence constraints and schedule the coflow.
We have the following lemma:
\begin{lem}\label{lem3:lem33}
For each coflow $k=1, \ldots, n$,
\begin{eqnarray*}
\hat{C}_{k}\leq 5\bar{C}_{k}.
\end{eqnarray*}
\end{lem}
\begin{proof}
Assume the last completed flow in coflow $k$ is $(i, j, k)$. Let $S = \left\{1, \ldots, k\right\}$ and $r_{max}(S)=\max_{f\in S} r_f$. Consider the schedule induced by the coflows $S$. Since all links $(i, j)$ in the network cores are busy from $r_{max}(S)$ to the start of the last completed flow in coflow $k$, we have
\begin{eqnarray}
\hat{C}_{k}   & \leq & r_{max}(S)  + \sum_{f\in S\setminus \left\{k\right\}} \left(L_{if}+L_{jf}\right)+L_{ik}+L_{jk} \label{lem3:eq1}\\
              & \leq & \bar{C}_{k} + \sum_{f\in S\setminus \left\{k\right\}} \left(L_{if}+L_{jf}\right)+L_{ik}+L_{jk} \label{lem3:eq2}\\
              & =    & \bar{C}_{k} + \sum_{f\in S} \left(L_{if}+L_{jf}\right) \label{lem3:eq3}\\
							& \leq & 5\bar{C}_{k} \label{lem3:eq4}
\end{eqnarray}
The inequality~(\ref{lem3:eq2}) is due to $\bar{C}_{k}\geq \bar{C}_{f}$ for all $f\in S$, we have $\bar{C}_{k}\geq r_{max}(S)$. The equation~(\ref{lem3:eq3}) shifts the coflow $k$ into the second and third terms. The inequality~(\ref{lem3:eq4}) is based on lemma~\ref{lem3:lem11} and lemma~\ref{lem3:lem22}.
\end{proof}

According to lemma~\ref{lem3:lem33}, we have the following theorem:
\begin{thm}\label{thm3:thm11}
The single-network-core-list-scheduling has an approximation ratio of, at most, $5\mu$.
\end{thm}
\begin{proof}
The proof is similar to that of theorem~\ref{thm:thm1}.
\end{proof}

We also have the following lemma:
\begin{lem}\label{lem3:lem44}
For each coflow $k=1, \ldots, n$,
\begin{eqnarray*}
\hat{C}_{k}\leq 4\bar{C}_{k}.
\end{eqnarray*}
when all coflows are released at time zero.
\end{lem}
\begin{proof}
The proof is similar to that of lemma~\ref{lem3:lem33}.
\end{proof}

According to lemma~\ref{lem3:lem44}, we have the following theorem:
\begin{thm}\label{thm3:thm22}
For the special case when all coflows are released at time zero, the single-network-core-list-scheduling has an approximation ratio of, at most, $4\mu$.
\end{thm}
\begin{proof}
The proof is similar to that of theorem~\ref{thm:thm1}.
\end{proof}

\section{Coflows of Multi-stage Jobs Scheduling Problem.}\label{sec:Algorithm4}
This section considers the coflows of multi-stage jobs scheduling problem. We can modify the linear programs~(\ref{coflow:main}), (\ref{incoflow:main}) and (\ref{incoflow:single:main}) to solve the problem of minimizing total weighted completion time and the problem of minimizing makespan. Then use the corresponding algorithm to schedule coflows. For example, we modify the linear programs~(\ref{coflow:main}). Let $\mathcal{T}$ be the set of jobs and $\mathcal{T}_{t}$ be the set of coflows belonging to the job $t$. We add an constraints~(\ref{job:coflow:b}) that ensures the completion time of any job is bounded by all its coflows. The goal is to minimize the total weighted completion time of a given set of multi-stage jobs. We can formulate our problem as the following linear programming relaxation:

\begin{subequations}\label{job:coflow:main}
\begin{align}
& \text{min}  && \sum_{t \in \mathcal{T}} w_{t} C_{t}     &   & \tag{\ref{job:coflow:main}} \\
& \text{s.t.} && (\ref{coflow:a})-(\ref{coflow:d}) &&  \notag \\
&  && C_{t} \geq C_{k},&& \forall t\in \mathcal{T}, \forall k\in \mathcal{T}_{t} \label{job:coflow:b} 
\end{align}
\end{subequations}

In the problem of minimizing makespan, we add an constraints~(\ref{job2:coflow:b}) that ensures the makespan is bounded by all coflows. We can formulate our problem as the following linear programming relaxation:

\begin{subequations}\label{job2:coflow:main}
\begin{align}
& \text{min}  && C_{max}     &   & \tag{\ref{job2:coflow:main}} \\
& \text{s.t.} && (\ref{coflow:a})-(\ref{coflow:d}) &&  \notag \\
&  && C_{max} \geq C_{k},&& \forall k\in \mathcal{K} \label{job2:coflow:b} 
\end{align}
\end{subequations}

We can analysis the last completed flow in each job and obtain the same result of previous analysis. Therefore, the following theorems can be obtained.
\begin{thm}\label{thm4:thm11}
The proposed algorithm has an approximation ratio of $O(\mu)$ for minimizing the total weighted completion time of a given set of multi-stage jobs.
\end{thm}

\begin{thm}\label{thm4:thm2}
The proposed algorithm has an approximation ratio of $O(\mu)$ for minimizing the makespan of a given set of multi-stage jobs.
\end{thm}

\section{Concluding Remarks}\label{sec:Conclusion}
This paper first develops two polynomial-time approximation algorithms to solve the coflow scheduling problem with precedence constraints in identical parallel networks. In the divisible coflow scheduling problem, the proposed algorithm achieves $(6-\frac{2}{m})\mu$ and $(5-\frac{2}{m})\mu$ approximate ratios for arbitrary release time and zero release time, respectively. In the indivisible coflow scheduling problem, the proposed algorithm achieves $(4m+1)\mu$ and $4m\mu$ approximate ratios for arbitrary release time and zero release time, respectively. This paper then considers the single network core scheduling problem. In single network core, this paper proposes an algorithm achieves $5\mu$ and $4\mu$ approximate ratios for arbitrary release time and zero release time, respectively.
This result represents an improvement upon the previous best approximation ratio of $O(\tilde{\mu} \log(N)/ \log(\log(N)))$ where $\tilde{\mu}$ is the maximum number of coflows in a job and $N$ is the number of servers.

\end{document}